\documentclass[11pt]{article}
\usepackage[margin=1in]{geometry}
\usepackage{microtype}
\usepackage{amsmath,amssymb,amsthm}
\usepackage{algorithm}
\usepackage{algpseudocode}
\usepackage{tikz}
\usepackage{pgfplots}
\pgfplotsset{compat=1.16}
\usepackage{booktabs}
\usepackage[numbers,sort&compress]{natbib}
\definecolor{navy}{RGB}{0,40,120}
\definecolor{citegreen}{RGB}{0,110,60}
\usepackage[colorlinks=true,linkcolor=navy,citecolor=citegreen,urlcolor=navy]{hyperref}
\newcommand{\doi}[1]{\href{https://doi.org/#1}{doi:\nolinkurl{#1}}}

\newtheorem{theorem}{Theorem}
\newtheorem{lemma}{Lemma}
\newtheorem{proposition}{Proposition}
\newtheorem{fact}{Fact}
\theoremstyle{definition}
\newtheorem{definition}{Definition}
\theoremstyle{plain}

\newcommand{\defeq}{\stackrel{\text{def}}{=}}
\newcommand{\E}{\mathbb{E}}

\newcommand{\cP}{\mathcal{P}}
\newcommand{\cX}{\mathcal{X}}
\newcommand{\ind}{\mathbf{1}}
\newcommand{\ALG}{\mathrm{ALG}}
\newcommand{\OPT}{\mathrm{OPT}}
\newcommand{\LB}{\mathrm{LB}}
\newcommand{\VN}{\mathrm{VN}}

\newcommand{\thnet}{\theta_{\mathrm{net}}}

\title{Stochastic Network Speed Scaling with Competitive Ratios\\ Independent of the Network}
\author{Yash Khanna\\[2pt]
  {\normalsize Independent Researcher, India}\\[1pt]
  {\normalsize\href{mailto:yashkhanna@alum.iisc.ac.in}{\texttt{yashkhanna@alum.iisc.ac.in}}}}
\date{}

\begin{document}

\maketitle

\begin{abstract}
In network speed scaling, jobs arrive over time at a network of servers whose speeds can be tuned, every job must be processed by the servers along one of its allowed routes, and the goal is to minimize the total flow time plus the total energy. For stochastic arrivals, the competitive ratio of Vaze and Nair depends on the network, through the lengths of the routes it uses. We show that this dependence can be removed: we give an algorithm, which routes the jobs by solving a convex program and runs every server at a fixed speed, whose competitive ratio depends only on the power functions; for $P(s)=s^2$, it is at most the golden ratio $\varphi\approx1.618$. The key idea is a lower bound on the optimal cost which, like the cost of our algorithm, is a sum over the servers of a function of each server's load, so the analysis reduces to a single server. We also show that the routing rule of Vaze and Nair can be a factor $\Omega(L)$ away from optimal, where $L$ is the length of the longest route.

\end{abstract}

\section{Introduction}

Energy consumption is a major concern for computing systems, from large data centers to battery-powered devices. A common technique to save energy is \emph{speed scaling}, also known as dynamic voltage and frequency scaling: a server, such as a processor, can run at different speeds, and running more slowly consumes much less power~\cite{Alb10}. In the standard model, running at speed $s$ consumes power $P(s)$, where $P$ is convex and increasing; a typical choice is $P(s)=s^\alpha$ for a constant $\alpha>1$, e.g., $\alpha=3$ by the cube-root rule for CMOS devices~\cite{Alb10}. Convexity means that doubling the speed more than doubles the power, so processing a given amount of work slowly uses less energy than processing it fast.

Speed scaling therefore creates a trade-off between the quality of service and the energy: running fast completes the jobs early but consumes a lot of energy, and running slowly saves energy but makes the jobs wait. A standard way to capture this trade-off, introduced by Albers and Fujiwara~\cite{AF07}, is to minimize the total \emph{flow time} of the jobs, i.e., the sum over all jobs of the time between their arrival and their completion, plus the total energy. For example, suppose that a single job with $w$ units of work is processed alone at a constant speed $s$, with $P(s)=s^2$. The job takes time $w/s$, so it contributes $w/s$ to the flow time, and the energy is $(w/s)\cdot P(s)=ws$. The total cost $w(1/s+s)$ is minimized at $s=1$, where it equals $2w$: a higher speed wastes energy, and a lower speed wastes time. When several jobs are present, the server should run faster, since every waiting job adds to the flow time. Starting with the work of Yao, Demers, and Shenker~\cite{YDS95}, speed scaling has been studied extensively for a single server and for parallel servers~\cite{AF07,BKP07,ALW10,VN20b}.

In practice, jobs arrive over time, and the speeds have to be chosen without knowing the future. Such algorithms are called \emph{online}, and they are evaluated by their \emph{competitive ratio}: an algorithm is $c$-competitive if, on every input, its cost is at most $c$ times the cost of an optimal solution, denoted by $\OPT$; the closer $c$ is to $1$, the better. There are two common ways to model the input. In the \emph{worst case} setting, the jobs arrive at arbitrary times, and $\OPT$ knows all the arrivals in advance. In the \emph{stochastic} setting, the jobs arrive at random times according to a known probabilistic model, the costs are long-run averages, and $\OPT$ is the best policy which, like the algorithm, does not know the future; an algorithm may thus use the arrival rates, but not the actual arrival times.

Many systems, such as flow shops and multi-hop communication networks, are better modelled as a \emph{network} of servers~\cite{VN20a}, where a job is complete only after it has been processed by all the servers on some route. Vaze and Nair~\cite{VN19,VN20a} initiated the study of this \emph{network speed scaling} problem, where an algorithm decides the route of every job, the scheduling at every server, and the speeds. In the \emph{stochastic} setting, where jobs arrive according to Poisson processes and bring exponentially distributed work to every server they visit, their competitive ratio for $P(s)=s^\alpha$ is $\theta_\alpha\cdot\thnet$, where $\theta_\alpha=3+\max\{1,2\alpha-2\}$, and $\thnet\geqslant1$ is a traffic-weighted ratio between the weights of the routes used by a flow and the smallest weight of its routes, where the weight of a route is a sum of per-server constants (its length, when all servers are identical); their routing is chosen by a convex program which only accounts for the energy. In the \emph{worst case} setting, where unit-size jobs arrive at arbitrary times, they gave constant competitive algorithms for a series of servers (a \emph{tandem}), with a ratio of $18$ for $P(s)=s^2$~\cite{VN19,VN20a}, and for parallel paths of equal length; for paths of different lengths, their competitive ratio depends on the path diversity, and for networks shared by several flows, on the maximum number of flows sharing a server~\cite{VN20a}.
This naturally raises the following question.
\begin{center}
\emph{Does the competitive ratio for stochastic network speed scaling have to depend on the network?}
\end{center}
In this work, we answer this question in the negative.

\subsection{Our Results}

Every job belongs to a \emph{flow}, and all the jobs of a flow have the same set of allowed \emph{routes}; e.g., in a communication network, a flow could consist of the packets from a given source to a given destination. A route is a sequence of servers (possibly with repetitions), and every job has to be processed along one of its allowed routes, one server after the other. In the stochastic setting, the jobs of every flow arrive according to a Poisson process, the standard model of memoryless arrivals at a given rate, and every visit of a job to a server brings an amount of work which is exponentially distributed; the formal model is in Section~\ref{sec:stoch}. Our main result is the following.

\begin{theorem}[Informal version of Theorems~\ref{thm:stoch} and~\ref{thm:phi}]\label{thm:informal-stoch}
In the stochastic setting, when every server has a power function $P_v(s)=\gamma_vs^\alpha$, there is an algorithm whose competitive ratio is at most $\kappa(\alpha)\leqslant\max\{e,\alpha-1\}$, where $\kappa(\alpha)$ depends only on $\alpha$, and $\kappa(2)=\varphi=(1+\sqrt5)/2\approx1.618$. For every $\varepsilon>0$, a version whose ratio is at most $(1+\varepsilon)\kappa(\alpha)$ runs in polynomial time. The guarantee holds for every network and for arbitrary sets of allowed routes.
\end{theorem}

If the exponents differ across servers, i.e., $P_v(s)=\gamma_vs^{\alpha_v}$, the same algorithm is $\max_v\kappa(\alpha_v)$-competitive (Theorem~\ref{thm:stoch}). The values of $\kappa(\alpha)$ for several $\alpha$ are listed in Appendix~\ref{app:kappa}.

We also show that the dependence on the network in~\cite{VN20a} is inherent to their routing rule, and is not merely an artifact of the analysis.

\begin{proposition}[Informal version of Proposition~\ref{prop:separation}]\label{prop:informal}
For every $L\geqslant4$ and $\varepsilon>0$, there is an instance with a single flow, whose longest route has $L$ servers, such that every policy which routes the jobs according to the routing rule of~\cite{VN20a} has cost at least $(1-\varepsilon)L/3$ times the optimal cost, irrespective of its scheduling and its speeds.
\end{proposition}

The instance behind Proposition~\ref{prop:informal} is shown in Figure~\ref{fig:instance}. A job can use a short route through $3$ servers, or one of $m$ long routes through $L$ servers each. Since the power function is convex, spreading the traffic over many routes lowers the load, and hence the energy, of every server; the routing rule of~\cite{VN20a}, which minimizes the energy alone, therefore sends most of the traffic over the long routes when $m$ is large. However, every server at which a job is processed costs something, even at a very low load: while the job is processed there, it contributes to the flow time, and the server consumes energy (at least $2$ per unit of work for $P(s)=s^2$, as in the example above). So at low loads, a long route costs about $L/3$ times as much as the short route, and it is better to send every job along the short route, which is what our algorithm does.

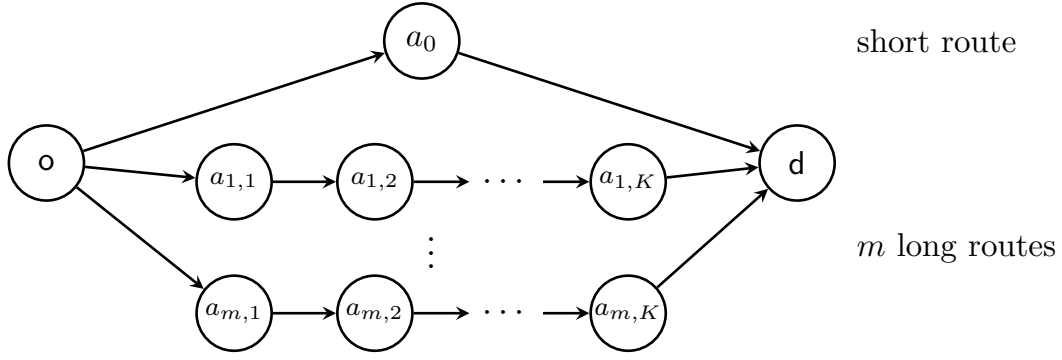
\begin{figure}[t]
\centering
\resizebox{0.85\linewidth}{!}{%
\begin{tikzpicture}[x=1cm,y=1cm,
  srv/.style={circle,draw,thick,minimum size=8mm,inner sep=0pt,font=\small},
  lnk/.style={->,>=stealth,thick}]
\node[srv] (s) at (0,0) {$\mathsf{o}$};
\node[srv] (d) at (8,0) {$\mathsf{d}$};
\node[srv] (a0) at (4,1.3) {$a_0$};
\draw[lnk] (s) -- (a0); \draw[lnk] (a0) -- (d);
\foreach \j/\y in {1/-0.2,m/-1.6}{
  \node[srv] (b1\j) at (2,\y) {\scriptsize $a_{\j,1}$};
  \node[srv] (b2\j) at (3.5,\y) {\scriptsize $a_{\j,2}$};
  \node (dots\j) at (4.9,\y) {$\cdots$};
  \node[srv] (bK\j) at (6.2,\y) {\scriptsize $a_{\j,K}$};
  \draw[lnk] (s) -- (b1\j); \draw[lnk] (b1\j) -- (b2\j); \draw[lnk] (b2\j) -- (dots\j);
  \draw[lnk] (dots\j) -- (bK\j); \draw[lnk] (bK\j) -- (d);
}
\node at (4.1,-0.85) {$\vdots$};
\node[font=\small,anchor=west] at (8.5,1.3) {short route};
\node[font=\small,anchor=west] at (8.5,-0.9) {$m$ long routes};
\end{tikzpicture}}
\caption{The instance of Propositions~\ref{prop:informal} and~\ref{prop:separation}: a single flow from $\mathsf{o}$ to $\mathsf{d}$ with one short route through $3$ servers and $m$ long routes through $L=K+2$ servers each.}
\label{fig:instance}
\end{figure}

We highlight two consequences of these results.
\begin{enumerate}
\item Firstly, our guarantee does not depend on the network, the routes, the number of flows, or the arrival rates. Our routing rule differs from that of~\cite{VN20a} only in the objective of a convex program, and every server runs at the best \emph{gated static} speed, i.e., the fixed speed, used only while the server has jobs, which minimizes its cost given its load.
\item Secondly, even when $\thnet=1$, our bound improves the factor $\theta_\alpha$ of~\cite{VN20a} ($5$ for $\alpha=2$, or $2$ with the speed rule of~\cite{VN20b}) to $\varphi$. Numerically, $\kappa(1.5)\approx1.325$, $\kappa(3)\approx2.148$, and $\kappa(10)\approx5.064$, while $\theta_\alpha=3+\max\{1,2\alpha-2\}$ equals $4$, $7$, and $21$, respectively.
\end{enumerate}

\subsection{Related Work}

\paragraph{Speed scaling on one or more servers.} Yao, Demers, and Shenker~\cite{YDS95} introduced speed scaling for jobs with deadlines, where the goal is to minimize the energy while completing every job by its deadline, and Bansal, Kimbrel, and Pruhs~\cite{BKP07} gave improved online algorithms for this problem. Albers and Fujiwara~\cite{AF07} introduced the flow time plus energy objective; see the survey~\cite{Alb10}. For a single server, SRPT with a speed based on the number of outstanding jobs is $2$-competitive for every power function~\cite{ALW10}. Very recently, Antoniadis et al.~\cite{AGHV26} showed that minimizing the total flow time plus energy offline is NP-hard, e.g., for unit-weight jobs of arbitrary sizes, while it becomes polynomial-time solvable when a completion-time ordering of the jobs is given. For multiple servers, Vaze and Nair studied SRPT with speed scaling~\cite{VN20b}, a sum power constraint~\cite{VN21a}, and MapReduce-type precedence constraints~\cite{9695966}; in the last setting, their competitive ratios do not depend on the number of jobs or servers. Bampis et al.~\cite{BKLLS18} gave a framework based on configuration linear programs and randomized rounding for several energy minimization problems, including power-aware job shops and the routing of demands in a network so as to minimize the power; see also~\cite{AFZZ12} for routing with speed-scalable links, and Gupta, Krishnaswamy, and Pruhs~\cite{GKP12}, who analyzed online routing with speed-scalable routers, where the goal is to minimize the sum of the $\alpha$-th powers of the edge loads, via online primal-dual techniques. These works on networks minimize the energy, possibly under deadlines, and do not model the flow time.

\paragraph{Stochastic models.} In the stochastic setting, the speed of a server is the service rate of a queue. Wierman, Andrew, and Tang~\cite{WAT09} compared static and dynamic speeds for a processor-sharing queue, and Andrew, Lin, and Wierman~\cite{ALW10} showed that, for a single server with a known mean workload, gated static speeds are nearly optimal. George and Harrison~\cite{GH01} showed how to compute the optimal state-dependent service rates of an M/M/1 queue (a single server with Poisson arrivals and exponentially distributed service times) with a cost of effort, and compared the resulting cost with that of the best static rate; see also~\cite{AS06}. Our algorithm uses static speeds, and our analysis bounds their loss against every dynamic policy on a network. In data centers, Baik and Ko~\cite{BK23} studied the joint control of workload routing and server speeds with a robust queueing-theoretic approach, and Chen, Xu, and Xu~\cite{CXX24} studied dynamic data center provisioning through speed scaling from a primal-dual perspective. For tandem networks with a bank of parallel servers at every stage, Vaze and Nair~\cite{VN19} showed that random routing with gated static speeds is constant competitive, and they extended this approach to general networks in~\cite{VN20a}.

\paragraph{Predictions and uncertainty.} Learning-augmented algorithms for speed scaling with deadlines were given in~\cite{BMRS20,AJS22}, and Balkanski et al.~\cite{BPSW23} gave a general framework for energy-efficient scheduling with predictions, which also covers flow time. Bampis et al.~\cite{BDKLP21} studied speed scaling when the workload of a job can only be learned by paying for a query.

\paragraph{Queueing networks.} We use Kelly's product-form theorem~\cite{Kel75,Kel79} (see~\cite{HB13} for an introduction). Choosing routes and static speeds resembles Kleinrock's capacity and flow assignment problem~\cite{Kle76}.

\subsection{Proof Overview}

Fix any policy and a server $v$ at which work arrives at a long-run rate $\rho_v$. There are two simple lower bounds on the cost at $v$, both used by Vaze and Nair~\cite{VN19,VN20a}: every unit of work processed at speed $\sigma$ costs at least $(1+P_v(\sigma))/\sigma$ (at least one job is present while $v$ is busy), which is at least $c^*_v\defeq\min_{\sigma>0}(1+P_v(\sigma))/\sigma$, the smallest possible cost of processing one unit of work at $v$ (e.g., $c^*_v=2$ for $P_v(s)=s^2$, attained at speed~$1$), and by Jensen's inequality the average power is at least $P_v(\rho_v)$. We combine them into a single bound $F_v(\rho_v)$, for a convex function $F_v$ with $F'_v(\rho)=\max\{c^*_v,P'_v(\rho)\}$ (Lemmas~\ref{lem:F} and~\ref{lem:node-lb}), and route by the resulting convex program. The routing program of~\cite{VN20a} minimizes $\sum_vP_v(\rho_v)$, whose marginal cost vanishes at low loads, so it does not ``see'' how many servers a route has (Proposition~\ref{prop:informal}). With static speeds, Kelly's theorem makes every server behave like an M/M/1 queue, whose average number of jobs and busy fraction have simple closed forms, whatever the topology, so the cost of our algorithm at $v$ is an explicit function $G_v(\rho_v)$. Thus $F_v$ is a lower bound for every policy and $G_v$ is the cost of our algorithm, both at a single server as a function of its load, and the competitive ratio is the largest ratio $G_v(\rho)/F_v(\rho)$.

\paragraph{Why the network plays no role.} The key ingredient is that three quantities split over the servers in the same way, each server contributing a function of its own load alone: the lower bound on the cost of every policy (Proposition~\ref{prop:lb}), the objective of our routing program, which is exactly this lower bound, and the cost of our algorithm (Lemma~\ref{lem:alg-stoch}, by Kelly's theorem). Hence our routing loses nothing with respect to the lower bound, and the ratio is at most $\max_v\sup_\rho G_v(\rho)/F_v(\rho)$, which depends only on the power functions. The network only determines which loads are possible, and it affects our algorithm and the lower bound in exactly the same way. In contrast, the routing program of~\cite{VN20a} optimizes the energy alone, while the flow time grows with the length of the routes; this mismatch between what the routing optimizes and what the jobs pay is what the factor $\thnet$ accounts for.

\section{Preliminaries}\label{sec:prelims}

\paragraph{Notation.} We use $[n]\defeq\{1,2,\ldots,n\}$. A \emph{route} is a finite sequence of servers $p=(v_1,\ldots,v_L)$, possibly with repetitions; $m_p(v)$ denotes the number of times $p$ visits server $v$, and $L$ is the \emph{length} of $p$. For functions $f,g$ on the same domain, we write $f\leqslant g$ if $f(t)\leqslant g(t)$ for all $t$. We use $\varphi\defeq(1+\sqrt5)/2$, which satisfies $\varphi-1=1/\varphi$.

\paragraph{Servers.} There is a finite set $V$ of servers. Server $v$ has a power function $P_v:[0,\infty)\to[0,\infty)$, which is convex, nondecreasing, and continuously differentiable, with $P_v(0)=0$ and $P_v(s)/s\to\infty$ as $s\to\infty$. At every time $t$, server $v$ runs at a speed $s_v(t)\geqslant0$, consuming power $P_v(s_v(t))$, and it can divide its speed among the jobs present at it. The main example is $P_v(s)=\gamma_vs^{\alpha_v}$ with $\gamma_v>0$ and $\alpha_v>1$.

\begin{fact}[Jensen's inequality]\label{fact:jensen}
For a probability measure $\mu$, a $\mu$-integrable function $g$, and a convex $f$, we have $f\left(\int g\,d\mu\right)\leqslant\int f(g)\,d\mu$.
\end{fact}

\section{Stochastic Network Speed Scaling}\label{sec:stoch}

\subsection{Model}

\begin{definition}[Stochastic network speed scaling]\label{def:stoch}
An instance is given as follows.
\begin{enumerate}
\item \emph{Flows.} There are $D$ flows. Flow $i\in[D]$ has a finite set $\cP_i$ of allowed routes, and its jobs arrive according to a Poisson process with rate $\lambda_i>0$, independently of the other flows.
\item \emph{Service.} A job of flow $i$ which is routed along $p=(v_1,\ldots,v_L)\in\cP_i$ has to be processed at $v_1$, then at $v_2$, and so on, and it departs after its processing at $v_L$. At every visit to a server, the job brings an amount of work which is exponentially distributed with mean $1$, independently of everything else.
\end{enumerate}
\end{definition}

The independence of the works at different visits (as in~\cite{VN19,VN20a}) makes the product-form analysis below possible; it is the main limitation of this setting, since, e.g., a packet has the same size at every hop.

A policy decides the route of every job, how every server shares its speed among the jobs present at it, and the speeds. Let $N(t)$ be the number of jobs in the network, and $N_v(t)$ the number of jobs located at server $v$ (waiting or in service), so that $N(t)\geqslant\sum_vN_v(t)$. Starting from an empty network, the cost of a policy $\pi$ is
\begin{equation}\label{eq:cost}
J(\pi)\defeq\limsup_{T\to\infty}\frac1T\int_0^T\Big(N(t)+\sum_{v\in V}P_v\big(s_v(t)\big)\Big)dt\,,
\end{equation}
the long-run average of the number of jobs plus the total power. This is a random variable, and all our statements about it hold almost surely, i.e., with probability~$1$. By Little's law~\cite{HB13}, $J(\pi)/\sum_i\lambda_i$ is the mean response time plus the mean energy per job, which is the objective of~\cite{VN20a}.

A \emph{routing} is a vector $x=(x_{i,p})$, where $x_{i,p}$ is the rate at which jobs of flow $i$ are routed along $p$; the set of routings is $\cX\defeq\{x\geqslant0:\sum_{p\in\cP_i}x_{i,p}=\lambda_i\ \forall i\}$. The \emph{load} of server $v$ is the rate at which jobs visit it, $\rho_v(x)\defeq\sum_i\sum_{p\in\cP_i}m_p(v)x_{i,p}$. Throughout, $\lambda$ denotes arrival rates of flows and $\rho$ denotes loads of servers. Following~\cite{VN19,VN20a}, who restrict attention to policies which admit long-run averages and do not use the work of a job before processing it, we compare against the following class. Its conditions only ask that the long-run routing rates exist and that every server keeps up with its work, which is what our lower bound needs. A policy $\pi$ is \emph{admissible} if, almost surely:
\begin{itemize}
\item[(A1)] for every flow $i$ and $p\in\cP_i$, the limit $x^\pi_{i,p}\defeq\lim_{T\to\infty}A_{i,p}(T)/T$ exists, where $A_{i,p}(T)$ is the number of jobs of flow $i$ arriving in $[0,T]$ and routed along $p$ (then $x^\pi\in\cX$ whenever $J(\pi)<\infty$, since jobs which are never routed stay in the network); and
\item[(A2)] for every server $v$, in the long run, work is processed at $v$ at least at the rate at which it arrives:
\[
\liminf_{T\to\infty}\frac1T\int_0^Ts_v(t)\,\ind\{N_v(t)\geqslant1\}\,dt\geqslant\rho_v(x^\pi)\,.
\]
\end{itemize}
For example, (A2) holds whenever $\pi$ does not use the work of a visit before processing it, $v$ receives $\rho_v(x^\pi)T+o(T)$ visits in $[0,T]$, and the work left at $v$ at time $T$ is $o(T)$. Indeed, the work $W_k$ of the $k$-th visit to $v$ is then independent of the history up to the start of that visit, so $\sum_{k\leqslant n}(W_k-1)$ is a martingale with increments of variance $1$, which is $o(n)$ almost surely; hence the work processed at $v$ by time $T$, which is at most $\int_0^Ts_v\ind\{N_v\geqslant1\}$, is $\rho_v(x^\pi)T-o(T)$. We use $\OPT$ to denote the infimum of the constants $c$ such that some admissible policy $\pi$ has $J(\pi)\leqslant c$ almost surely, and the competitive ratio of an admissible policy $A$ is $J(A)/\OPT$. Vaze and Nair~\cite{VN19,VN20a} use stationary expectations instead of almost-sure time averages; the two agree for Algorithm~\ref{alg:stoch} (Lemma~\ref{lem:alg-stoch}). As in their work, policies which know the work of a visit in advance are not covered. On the other hand, (A1) and (A2) do not restrict what a policy knows about the arrivals, so the lower bound of Proposition~\ref{prop:lb}, and hence our competitive ratio, also holds against policies which know all the future arrival times in advance.

\begin{fact}[Kelly's theorem~\cite{Kel75,Kel79}]\label{fact:kelly}
Consider an open network of single-server first-come-first-served (FCFS) nodes, with finitely many classes of jobs. Jobs of class $c$ arrive according to independent Poisson processes with rate $a_c$ and follow a fixed route $r_c$. At node $v$, every service time is exponential with rate $\mu_v$, independently of everything else. If $\rho_v\defeq\sum_ca_cm_{r_c}(v)<\mu_v$ for every $v$, then the network is a positive recurrent Markov process, and in its stationary distribution $\Pr[N_v=n]=(1-\rho_v/\mu_v)(\rho_v/\mu_v)^n$; hence $\E[N_v]=\rho_v/(\mu_v-\rho_v)$ and $\Pr[N_v\geqslant1]=\rho_v/\mu_v$. Moreover, the network regenerates whenever it becomes empty, so by the renewal--reward theorem~\cite{Asm87}, the time averages in~\eqref{eq:cost} converge almost surely to the stationary expectations, which are finite.
\end{fact}

\subsection{A Lower Bound on the Cost of Any Policy}

\begin{definition}\label{def:F}
For a server $v$ and $\sigma>0$, let $h_v(\sigma)\defeq(1+P_v(\sigma))/\sigma$, let $s^*_v$ be the smallest minimizer of $h_v$, and $c^*_v\defeq h_v(s^*_v)$. Let $F_v(0)\defeq0$ and $F_v(\rho)\defeq\rho\cdot\min_{\sigma\geqslant\rho}h_v(\sigma)$ for $\rho>0$.
\end{definition}

If $v$ processes a single job at a constant speed $\sigma$, then every unit of work takes time $1/\sigma$, during which the job contributes $1/\sigma$ to the flow time, and $v$ consumes energy $P_v(\sigma)/\sigma$. Thus $c^*_v$ is the smallest possible cost of processing one unit of work at $v$. In words, $F_v(\rho)$ is a lower bound on the cost per unit time at $v$ (the number of jobs at $v$ plus the power of $v$) of any policy which sends work to $v$ at rate $\rho$; this is Lemma~\ref{lem:node-lb} below.

\begin{lemma}\label{lem:F}
The minimum of $h_v$ is attained, $c^*_v=P'_v(s^*_v)$, and $F_v(\rho)=c^*_v\rho$ for $\rho\leqslant s^*_v$ and $F_v(\rho)=1+P_v(\rho)$ for $\rho\geqslant s^*_v$. Moreover, $F_v$ is convex, nondecreasing, and continuously differentiable with $F'_v(\rho)=\max\{c^*_v,P'_v(\rho)\}$.
\end{lemma}
\begin{proof}
We drop the subscript $v$. Let $q(\sigma)\defeq\sigma P'(\sigma)-P(\sigma)$. For $\sigma_1<\sigma_2$, convexity gives $P(\sigma_2)-P(\sigma_1)\leqslant(\sigma_2-\sigma_1)P'(\sigma_2)$, so $q(\sigma_2)-q(\sigma_1)\geqslant\sigma_1(P'(\sigma_2)-P'(\sigma_1))\geqslant0$, i.e., $q$ is nondecreasing. Also $h'(\sigma)=(q(\sigma)-1)/\sigma^2$. Since $h(\sigma)\geqslant1/\sigma$ and $h(\sigma)\geqslant P(\sigma)/\sigma$, $h\to\infty$ at $0$ and at $\infty$, so its minimum is attained at some $s^*$ with $q(s^*)=1$, i.e., $s^*P'(s^*)=1+P(s^*)$, which gives $P'(s^*)=h(s^*)=c^*$. As $q$ is nondecreasing, $h$ is nonincreasing on $(0,s^*]$ and nondecreasing on $[s^*,\infty)$. Hence $\min_{\sigma\geqslant\rho}h(\sigma)$ equals $c^*$ if $\rho\leqslant s^*$, and $h(\rho)$ otherwise, which gives the formula for $F$. Its derivative is $c^*$ on $(0,s^*)$ and $P'$ on $(s^*,\infty)$; both equal $c^*$ at $s^*$, and since $P'$ is nondecreasing with $P'(s^*)=c^*$, we get $F'=\max\{c^*,P'\}$, which is continuous, positive, and nondecreasing.
\end{proof}

For example, for $P(s)=s^2$, we have $s^*=1$, $c^*=2$, and $F(\rho)=2\rho$ for $\rho\leqslant1$ and $F(\rho)=1+\rho^2$ for $\rho\geqslant1$. The formula $F'_v=\max\{c^*_v,P'_v\}$ says that the marginal cost of pushing traffic through $v$ is the larger of the cost $c^*_v$ of processing a job there and the marginal power; so at low loads, a route with $L$ servers costs about $L$ times as much as a route with one server.

The two regimes of $F_v$ correspond to the two lower bounds of Vaze and Nair~\cite[Lemmas~25 and~26]{VN20a} (see also~\cite[Lemmas~14 and~15]{VN19}): the cost $c^*_v\rho$ of processing all the work in isolation ($c^*_v$ is the ``hop weight'' of~\cite{VN20a}), and the power $P_v(\rho)$. $F_v$ is at least their maximum, and larger above $s^*_v$; with their maximum, the proof of Theorem~\ref{thm:phi} would give $2$ instead of $\varphi$. More importantly, we use $F_v$ inside the routing program, whereas that of~\cite{VN20a} only uses $P_v$.

\begin{lemma}[Lower bound at a server]\label{lem:node-lb}
For every admissible policy $\pi$ and server $v$, almost surely,
\[
\liminf_{T\to\infty}\frac1T\int_0^T\Big(N_v(t)+P_v\big(s_v(t)\big)\Big)dt\geqslant F_v\big(\rho_v(x^\pi)\big)\,.
\]
\end{lemma}
\begin{proof}
The idea is that, among all ways of processing the same work while busy for the same amount of time, a constant speed uses the least energy. Fix a realization and $T>0$. Let $b(t)\defeq\ind\{N_v(t)\geqslant1\}$, $B_T\defeq\frac1T\int_0^Tb$, and $W_T\defeq\frac1T\int_0^Ts_vb$ be the fraction of time $v$ is busy and the average rate at which it processes work; note $B_T\leqslant1$. If $W_T=0$, then $\frac1T\int_0^T(N_v+P_v(s_v))\geqslant0=F_v(W_T)$. Otherwise, let $\mu$ be the uniform probability measure on the busy times. Then
\begin{alignat*}{2}
\frac1T\int_0^T\Big(N_v+P_v(s_v)\Big)dt&\geqslant\frac1T\int_0^T\Big(b+P_v(s_v)\,b\Big)dt&\;&(\because N_v\geqslant b,\ P_v\geqslant0)\\
&\geqslant B_T+B_T\,P_v\left(W_T/B_T\right)&&(\text{by Fact~\ref{fact:jensen}})\\
&=W_T\cdot h_v\left(W_T/B_T\right)\geqslant F_v(W_T)\,.&&(\because W_T/B_T\geqslant W_T)
\end{alignat*}
Since $F_v$ is continuous and nondecreasing (Lemma~\ref{lem:F}), and $\liminf_TW_T\geqslant\rho_v(x^\pi)$ by (A2), the claim follows.
\end{proof}

Summing Lemma~\ref{lem:node-lb} over the servers (using $N\geqslant\sum_vN_v$, and that the $\limsup$ of a sum is at least the sum of the $\liminf$s) gives the following.

\begin{proposition}\label{prop:lb}
For every admissible policy $\pi$, almost surely,
\[
J(\pi)\geqslant\sum_{v\in V}F_v\big(\rho_v(x^\pi)\big)\geqslant\LB\defeq\min_{x\in\cX}\ \sum_{v\in V}F_v\big(\rho_v(x)\big)\,.
\]
\end{proposition}

The program defining $\LB$ is convex, since $\cX$ is a polytope, $\rho_v(x)$ is linear, and every $F_v$ is convex; given oracles for $F_v$ and $F'_v$, an $x$ with $\sum_vF_v(\rho_v(x))\leqslant(1+\varepsilon)\LB$ can be computed in time polynomial in the number of servers, the number of allowed routes, and $\log(1/\varepsilon)$, e.g., by the ellipsoid method or interior point methods. For $P_v(s)=\gamma_vs^{\alpha_v}$, these oracles are explicit, since the functions $F_v$ are: the condition $s^*_vP'_v(s^*_v)=1+P_v(s^*_v)$ from the proof of Lemma~\ref{lem:F} gives $s^*_v=(\gamma_v(\alpha_v-1))^{-1/\alpha_v}$ and $c^*_v=\alpha_v\gamma_v(s^*_v)^{\alpha_v-1}$.

\subsection{The Algorithm}

For a server $v$, $\rho>0$, and $s>\rho$, let
\[
g_v(\rho,s)\defeq\frac{\rho}{s-\rho}+\frac{\rho P_v(s)}{s}\,,\qquad G_v(\rho)\defeq\min_{s>\rho}g_v(\rho,s)\,,\qquad G_v(0)\defeq0\,,
\]
and let $s_v(\rho)$ be a minimizer (it exists, since $g_v(\rho,s)\to\infty$ as $s\downarrow\rho$ and as $s\to\infty$). By Fact~\ref{fact:kelly}, $g_v(\rho,s)$ is the cost of an M/M/1 queue with arrival rate $\rho$ which runs at speed $s$ when it is non-empty and at speed $0$ otherwise: the mean number of jobs plus the busy fraction $\rho/s$ times the power $P_v(s)$. In words, $G_v(\rho)$ is the cost per unit time at $v$ of our algorithm when the load of $v$ is $\rho$ (Lemma~\ref{lem:alg-stoch}).

\begin{algorithm}[h]
\caption{Static routing with gated static speeds.}\label{alg:stoch}
\begin{algorithmic}[1]
\State Compute an optimal routing $x^*$ for $\LB$, and let $\rho^*_v\defeq\rho_v(x^*)$ for every $v\in V$.
\State Route every arriving job of flow $i$ along $p\in\cP_i$ with probability $x^*_{i,p}/\lambda_i$, independently.
\State Every server $v$ with $\rho^*_v>0$ processes its jobs in FCFS order at speed $s_v(\rho^*_v)$ when it is non-empty, and at speed $0$ otherwise (servers with $\rho^*_v=0$ are never visited).
\end{algorithmic}
\end{algorithm}

\begin{lemma}\label{lem:alg-stoch}
Algorithm~\ref{alg:stoch} is admissible, and $J(\ALG)=\sum_vG_v(\rho^*_v)$ almost surely.
\end{lemma}
\begin{proof}
Since every job picks its route independently, the jobs of flow $i$ routed along $p$ form independent Poisson processes (the thinning property) with rates $x^*_{i,p}$; we treat each of them as a class with the fixed route $p$. At server $v$, every visit brings exponential work of mean $1$ and $v$ runs at speed $s_v\defeq s_v(\rho^*_v)>\rho^*_v$ while busy, so every service time is exponential with rate $s_v$. By Fact~\ref{fact:kelly}, the network is positive recurrent and, almost surely, the long-run average number of jobs at $v$ is $\rho^*_v/(s_v-\rho^*_v)$, and its long-run average power is $\Pr[N_v\geqslant1]\,P_v(s_v)=\rho^*_vP_v(s_v)/s_v$. Since $N=\sum_vN_v$ here, $J(\ALG)=\sum_vg_v(\rho^*_v,s_v)=\sum_vG_v(\rho^*_v)$. Condition (A1) holds with $x^\ALG=x^*$ by the strong law of large numbers, and (A2) holds since the long-run rate at which work is processed at $v$ is $s_v\Pr[N_v\geqslant1]=\rho^*_v$.
\end{proof}

\begin{theorem}\label{thm:stoch}
Let $\kappa_v\defeq\sup_{\rho>0}G_v(\rho)/F_v(\rho)$ and $\kappa^*\defeq\max_v\kappa_v$. For every admissible policy $\pi$, almost surely $J(\ALG)\leqslant\kappa^*\cdot J(\pi)$; i.e., Algorithm~\ref{alg:stoch} is $\kappa^*$-competitive.
\end{theorem}
\begin{proof}
We have
\begin{align*}
J(\ALG)&=\sum_vG_v(\rho^*_v)&&(\text{by Lemma~\ref{lem:alg-stoch}})\\
&\leqslant\kappa^*\sum_vF_v(\rho^*_v)=\kappa^*\cdot\LB&&(\because G_v(0)=F_v(0)=0,\ x^*\text{ is optimal})\\
&\leqslant\kappa^*\cdot J(\pi)\,.&&(\text{by Proposition~\ref{prop:lb}})
\end{align*}
\end{proof}

If $\LB$ is only solved up to a factor $1+\varepsilon$, the same proof gives the ratio $(1+\varepsilon)\kappa^*$. For general power functions, we do not know whether $\kappa_v$ is always finite; the next lemma shows that it is finite for polynomial ones.

\begin{lemma}\label{lem:kappa}
Let $P_v(s)=\gamma s^\alpha$ with $\gamma>0$ and $\alpha>1$. Then the following hold.\par
\noindent\emph{Part 1.} $\kappa_v$ does not depend on $\gamma$; we denote it by $\kappa(\alpha)$.\par
\noindent\emph{Part 2.} $G_v\geqslant F_v$, and $G_v(\rho)/F_v(\rho)\to1$ as $\rho\to0$ and as $\rho\to\infty$.\par
\noindent\emph{Part 3.} $\kappa(\alpha)\leqslant\max\{e,\alpha-1\}$; in particular, $\kappa(\alpha)<\infty$.
\end{lemma}
\begin{proof}
\emph{Part 1.} Substituting $\rho=\gamma^{-1/\alpha}\ell$, $s=\gamma^{-1/\alpha}u$, and $\sigma=\gamma^{-1/\alpha}\tau$ turns $g_v(\rho,s)$ and $\rho h_v(\sigma)$ into the corresponding expressions for $\gamma=1$, and the constraints $s>\rho$ and $\sigma\geqslant\rho$ into $u>\ell$ and $\tau\geqslant\ell$, so we may assume $\gamma=1$. Then $(s^*)^{-\alpha}=\alpha-1$ and $c^*=P'(s^*)=\alpha(s^*)^{\alpha-1}$ (Lemma~\ref{lem:F}).

\emph{Part 2.} $G(\rho)\geqslant F(\rho)$ by Lemma~\ref{lem:node-lb} applied to a single server running at a gated static speed. For $\rho\leqslant s^*$, taking $s=\rho+s^*$ gives
\[
\frac{G(\rho)}{F(\rho)}\leqslant\frac{\rho/s^*+\rho(\rho+s^*)^{\alpha-1}}{\alpha(s^*)^{\alpha-1}\rho}\ \xrightarrow{\ \rho\to0\ }\ \frac{(s^*)^{-\alpha}+1}{\alpha}=1\,.
\]
For $\rho\geqslant s^*$ and any $\delta>0$, taking $s=(1+\delta)\rho$ gives $G(\rho)\leqslant1/\delta+(1+\delta)^{\alpha-1}\rho^\alpha$, while $F(\rho)=1+\rho^\alpha$, so $\limsup_{\rho\to\infty}G/F\leqslant(1+\delta)^{\alpha-1}$ for every $\delta>0$.

\emph{Part 3.} We use $1/s^*=(\alpha-1)(s^*)^{\alpha-1}$ and $(1+\frac1{\alpha-1})^{\alpha-1}\leqslant e$. For $\rho\leqslant s^*$, taking $s=\rho+s^*/(\alpha-1)$ gives
\begin{align*}
G(\rho)&\leqslant\frac{(\alpha-1)\rho}{s^*}+\rho\Big(\rho+\frac{s^*}{\alpha-1}\Big)^{\alpha-1}\\
&\leqslant\big((\alpha-1)^2+e\big)(s^*)^{\alpha-1}\rho=\frac{(\alpha-1)^2+e}{\alpha}\,F(\rho)\,,
\end{align*}
and $((\alpha-1)^2+e)/\alpha$ is at most $e$ if $\alpha-1\leqslant e$, and at most $\alpha-1$ otherwise. For $\rho\geqslant s^*$, taking $\delta=1/(\alpha-1)$ above gives $G(\rho)\leqslant(\alpha-1)+e\rho^\alpha\leqslant\max\{e,\alpha-1\}(1+\rho^\alpha)=\max\{e,\alpha-1\}F(\rho)$.
\end{proof}

The values of $\kappa(\alpha)$ in Section~1 were computed numerically: $s(\rho)$ is the unique root $s>\rho$ of $(\alpha-1)s^{\alpha-2}(s-\rho)^2=1$ (the left side increases from $0$ to $\infty$ on $s>\rho$), and by Part~2 of Lemma~\ref{lem:kappa}, it suffices to search a bounded range of $\rho$. The bound of Part~3 gives, e.g., $e\approx2.72$ for $\alpha=2$ and $9$ for $\alpha=10$, compared with $\kappa(2)=\varphi$ and $\kappa(10)\approx5.064$. For $\alpha=2$, $\kappa(2)$ has a closed form.

\begin{theorem}\label{thm:phi}
If $P_v(s)=\gamma_vs^2$ for every server $v$, then Algorithm~\ref{alg:stoch} runs $v$ at speed $\rho^*_v+\gamma_v^{-1/2}$ when it is non-empty, and it is $\varphi$-competitive.
\end{theorem}
\begin{proof}
By Part~1 of Lemma~\ref{lem:kappa}, we may assume $\gamma_v=1$. Then $g(\rho,s)=\rho/(s-\rho)+\rho s$ is convex in $s$, and $\partial g/\partial s=0$ at $s=\rho+1$, so $G(\rho)=\rho^2+2\rho$. If $\rho\leqslant1$, then $G(\rho)/F(\rho)=(\rho+2)/2\leqslant3/2<\varphi$. If $\rho\geqslant1$, then
\[
\varphi F(\rho)-G(\rho)=(\varphi-1)\rho^2-2\rho+\varphi=\frac1\varphi\left(\rho^2-2\varphi\rho+\varphi^2\right)=\frac{(\rho-\varphi)^2}\varphi\geqslant0\,,
\]
using $\varphi-1=1/\varphi$, with equality at $\rho=\varphi$. Hence $\kappa(2)=\varphi$, and Theorem~\ref{thm:stoch} completes the proof.
\end{proof}

For $P(s)=s^2$, the ratio $G/F$ is close to $1$ at very low and very high loads, where one of the two lower bounds is almost tight, and it is largest at intermediate loads, where both the flow time and the energy matter; its maximum $\varphi$ is attained at $\rho=\varphi$. Moreover, if $P_v(s)=s^2$ for every server and the routing of~\cite{VN20a} is also optimal for $\LB$ (e.g., if every flow has a single allowed route, or in the tandem networks of~\cite{VN19} with identical servers in every layer), then Algorithm~\ref{alg:stoch} can use the same routing, and it coincides with the algorithm of~\cite{VN20a} with the speed rule of~\cite{VN20b}; so Theorem~\ref{thm:phi} improves the bound of $2$ that~\cite[Theorem~27]{VN20a} gives for this algorithm to $\varphi$.

\subsection{The Routing Rule of Vaze and Nair can Lose a Factor \texorpdfstring{$\Omega(L)$}{Omega(L)}}

The algorithm of~\cite[Section~6]{VN20a} routes according to an optimal solution $x^\VN$ of the program $\min_{x\in\cX}\sum_vP_v(\rho_v(x))$~\cite[Eq.~(24)]{VN20a}.
Consider the following instance (Figure~\ref{fig:instance}) with $P_v(s)=s^2$ for every server. Let $K\geqslant2$, $m\geqslant1$, and $L\defeq K+2$. There is a single flow with rate $\lambda\in(0,1)$ and $m+1$ allowed routes: a \emph{short} route $p_0=(\mathsf{o},a_0,\mathsf{d})$, and $m$ \emph{long} routes $p_j=(\mathsf{o},a_{j,1},\ldots,a_{j,K},\mathsf{d})$ of length $L$, where the servers $a_0$ and $a_{j,k}$ are distinct. Writing $x_j$ for the rate routed along $p_j$, the objective of~\cite{VN20a} is $2\lambda^2+x_0^2+K\sum_{j\geqslant1}x_j^2$, and its unique minimizer satisfies $x_0=Kx_j$ (at the optimum, the derivatives $2x_0$ and $2Kx_j$ of the objective are equal), i.e.,
\begin{equation}\label{eq:xvn}
x^\VN_0=\frac{K\lambda}{K+m}\qquad\text{and}\qquad x^\VN_j=\frac{\lambda}{K+m}\quad\forall j\in[m]\,.
\end{equation}
\begin{proposition}[Formal version of Proposition~\ref{prop:informal}]\label{prop:separation}
On this instance, every admissible policy $\pi$ with $x^\pi=x^\VN$ satisfies, almost surely and irrespective of its scheduling and speeds,
\[
J(\pi)\geqslant\frac{2}{3(\lambda+2)}\left(2+\frac{K(m+1)}{K+m}\right)\OPT\,.
\]
The factor tends to $(K+2)/3=L/3$ as $\lambda\to0$ and $m\to\infty$; in particular, for every $\varepsilon>0$, it is at least $(1-\varepsilon)L/3$ when $\lambda$ is small and $m$ is large enough.
\end{proposition}
\begin{proof}
Consider the policy which routes every job along $p_0$ and runs $\mathsf{o},a_0,\mathsf{d}$ in FCFS order at speed $\lambda+1$ when non-empty. Each of these servers has load $\lambda$, so, as in Lemma~\ref{lem:alg-stoch}, the policy is admissible and its cost is $3g(\lambda,\lambda+1)=3G(\lambda)=3\lambda(\lambda+2)$ (see the proof of Theorem~\ref{thm:phi}). So $\OPT\leqslant3\lambda(\lambda+2)$. On the other hand, all the loads under $x^\VN$ are at most $\lambda<s^*=1$, so $F(\rho_v)=2\rho_v$ for every server, and by Proposition~\ref{prop:lb} and~\eqref{eq:xvn},
\[
J(\pi)\geqslant2\left(2\lambda+x^\VN_0+Km\,x^\VN_1\right)=2\lambda\left(2+\frac{K(m+1)}{K+m}\right),
\]
which gives the claim since $\OPT\leqslant3\lambda(\lambda+2)$.
\end{proof}

Here $\thnet\to L/3$ as $m\to\infty$, so the bound $\theta_\alpha\cdot\thnet$ of~\cite{VN20a} is tight for their routing rule up to the factor $\theta_\alpha$. In contrast, Algorithm~\ref{alg:stoch} routes every job along $p_0$, and is therefore optimal up to the factor $(\lambda+2)/2$ on this instance. Indeed, every load is at most $\lambda<1=s^*$ under every routing, so $F'(\rho_v)=2$ at every server, and the marginal cost of a route in the program defining $\LB$ is twice its length: $6$ for $p_0$ and $2L\geqslant8$ for every long route. Hence the unique optimal routing for $\LB$ sends all the traffic along $p_0$, and $J(\ALG)=3G(\lambda)=3\lambda(\lambda+2)$, while $\LB=6\lambda$.

\section{Conclusion and Open Problems}\label{sec:conclusion}

We showed that, for stochastic arrivals, the competitive ratio for network speed scaling does not have to depend on the network: routing by a convex program whose objective is a lower bound on the cost of every policy, together with a static speed at every server, gives a ratio which depends only on the power functions, and which is at most $\varphi$ for $P(s)=s^2$. We conclude with three open problems.
\begin{enumerate}
\item \emph{Tightness.} We do not know how far $\LB$ is from $\OPT$. What is the competitive ratio of Algorithm~\ref{alg:stoch} against $\OPT$, and can speeds which depend on the number of jobs at a server beat $\varphi$ for $P(s)=s^2$? Even for a single server, this requires comparing static and dynamic speeds~\cite{GH01,ALW10}.
\item \emph{Beyond independent exponential work.} What happens with general work distributions, or when a job has the same work at every server, as for packets in a communication network?
\item \emph{Arbitrary arrivals.} For the worst case setting, is there a network-independent competitive ratio achieved by a simple speed rule which, as in~\cite{VN19,VN20a}, does not use the works of the jobs?
\end{enumerate}

\section*{Acknowledgements}
YK thanks Rahul Vaze for introducing him to this problem during a virtual summer internship at the Tata Institute of Fundamental Research (TIFR) in 2020. The work was not completed during the internship, and the results in this paper were obtained later. During the internship, YK was a student at IISc and was supported by the Ministry of Education, Government of India.

\section*{Use of AI Tools}
During the preparation of this work, YK used Claude Opus 5.5 (Anthropic) to develop the results, to search the literature, and to write and typeset the manuscript. After using this tool, the author reviewed and edited the content as needed and takes full responsibility for the content of the publication.

\bibliographystyle{alpha}
\bibliography{mybibliography}

\appendix
\section{Numerical Values of \texorpdfstring{$\kappa(\alpha)$}{kappa(alpha)}}\label{app:kappa}

For $P(s)=s^\alpha$, Table~\ref{tab:kappa} lists $\kappa(\alpha)=\sup_{\rho>0}G(\rho)/F(\rho)$ (Lemma~\ref{lem:kappa}), the load $\rho$ at which the supremum is attained, the bound $\max\{e,\alpha-1\}$ of Part~3 of Lemma~\ref{lem:kappa}, and the factor $\theta_\alpha=3+\max\{1,2\alpha-2\}$ of~\cite{VN20a}. By Part~1 of Lemma~\ref{lem:kappa}, the values do not depend on $\gamma$. Figure~\ref{fig:ratio} shows $G/F$ for $\alpha=2$, where it has the closed form $(\rho+2)/2$ for $\rho\leqslant1$ and $(\rho^2+2\rho)/(1+\rho^2)$ for $\rho\geqslant1$ (Theorem~\ref{thm:phi}).

\begin{table}[h]
\centering
\begin{tabular}{@{}rcccc@{}}
\toprule
$\alpha$ & $\kappa(\alpha)$ & maximizing $\rho$ & $\max\{e,\alpha-1\}$ & $\theta_\alpha$ \\
\midrule
$1.25$ & $1.167$ & $5.628$ & $2.718$ & $4$ \\
$1.5$  & $1.325$ & $2.786$ & $2.718$ & $4$ \\
$2$    & $1.618$ & $1.618$ & $2.718$ & $5$ \\
$2.5$  & $1.891$ & $1.300$ & $2.718$ & $6$ \\
$3$    & $2.148$ & $1.163$ & $2.718$ & $7$ \\
$4$    & $2.630$ & $1.049$ & $3$ & $9$ \\
$5$    & $3.080$ & $1.004$ & $4$ & $11$ \\
$10$   & $5.064$ & $0.961$ & $9$ & $21$ \\
\bottomrule
\end{tabular}
\caption{The competitive ratio $\kappa(\alpha)$ of Algorithm~\ref{alg:stoch} for $P(s)=\gamma s^\alpha$, compared with the bound of Lemma~\ref{lem:kappa} and the factor $\theta_\alpha$ of~\cite{VN20a}.}
\label{tab:kappa}
\end{table}

\begin{figure}[h]
\centering
\begin{tikzpicture}
\begin{axis}[width=0.7\linewidth,height=5.2cm,xmode=log,xmin=0.01,xmax=100,ymin=0.95,ymax=1.7,
  xlabel={load $\rho$},ylabel={$G(\rho)/F(\rho)$},
  ytick={1,1.2,1.4,1.618},yticklabels={$1$,$1.2$,$1.4$,$\varphi$},
  grid=major,grid style={gray!25},thick]
\addplot[blue,domain=0.01:1,samples=60] {(x+2)/2};
\addplot[blue,domain=1:100,samples=200] {(x^2+2*x)/(1+x^2)};
\addplot[dashed,gray,domain=0.01:100] {1.6180339887};
\addplot[only marks,mark=*,mark size=1.6pt] coordinates {(1.6180339887,1.6180339887)};
\end{axis}
\end{tikzpicture}
\caption{The ratio $G/F$ for $P(s)=s^2$. It tends to $1$ at low and at high loads, and its maximum $\varphi$ is attained at $\rho=\varphi$.}
\label{fig:ratio}
\end{figure}
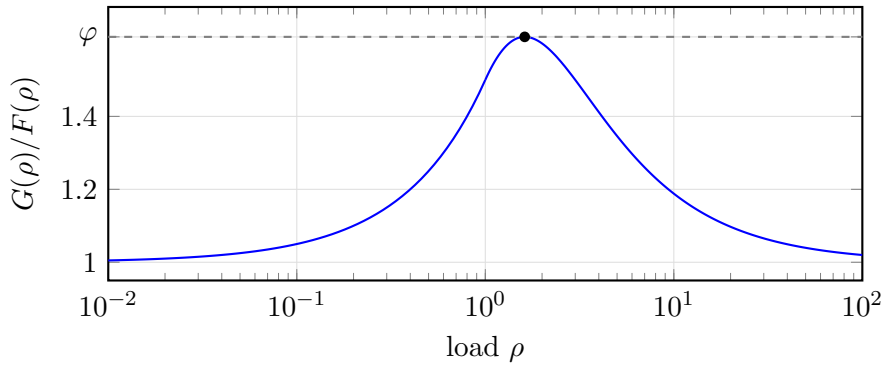

The values in Table~\ref{tab:kappa} were computed as follows. For $\rho>0$, the speed $s(\rho)=\rho+u$ minimizing $g(\rho,\cdot)$ is given by the unique root $u>0$ of $(\alpha-1)(\rho+u)^{\alpha-2}u^2=1$, found by bisection on $\log u$ (the left side is increasing in $u$), and $G(\rho)=\rho/u+\rho(\rho+u)^{\alpha-1}$; $F$ is given by Lemma~\ref{lem:F}, with $s^*=(\alpha-1)^{-1/\alpha}$ and $c^*=\alpha(s^*)^{\alpha-1}$. The ratio $G/F$ was maximized over $\rho\in[e^{-12},e^{12}]$ by a grid search on $\log\rho$ followed by a local refinement; by Part~2 of Lemma~\ref{lem:kappa}, the ratio tends to $1$ outside this range. 

\end{document}